\documentclass[a4paper,11pt]{article}
\usepackage{amssymb,palatino,amsthm,amsmath}
\usepackage{color}
\newtheorem{theorem}{Theorem}
\newtheorem{corollary}[theorem]{Corollary}
\newtheorem{lemma}[theorem]{Lemma}

\newtheorem{proposition}[theorem]{Proposition}
\newtheorem{example}[theorem]{Example}
\newtheorem{remark}[theorem]{Remark}
\newtheorem{remarks}[theorem]{Remarks}
\newcommand{\FF}{{\mathbb F}}
\newcommand{\K}{{\mathcal K}}
\newcommand{\KK}{{\mathbb K}}
\newcommand{\G}{{\mathcal G}}
\newcommand{\M}{{\mathcal M}}
\DeclareMathOperator{\rk}{rank}
\DeclareMathOperator{\ch}{char}

\begin{document}

\title{Computing the dimension of ideals in group algebras, \\ with an application to coding theory} 

\author{Michele Elia\thanks{Politecnico di Torino, Italy; e-mail: michele.elia@polito.it},~
Elisa Gorla\thanks{University of Neuch\^atel, Switzerland; e-mail: elisa.gorla@unine.ch.} }

\maketitle

\thispagestyle{empty}

\begin{abstract}
\noindent
The problem of computing the dimension of a left/right ideal 
in a group algebra $\FF[\G]$ of a finite group $\G$ over a field $\FF$ is considered. 
The ideal dimension is related to the rank of a matrix 
originating from a regular left/right representation of $\G$;
in particular, when $\FF[\G]$ is semisimple,  the dimension of a principal ideal is equal to the rank of the 
matrix representing a generator. 
From this observation, a bound and an efficient algorithm to compute the dimension of an ideal in a group ring are established. 
Since group codes are ideals in finite group rings, the algorithm allows 
efficient computation of their dimension.
\end{abstract}


\vspace{2mm}

\vspace{8mm}

\section{Introduction and preliminaries}\label{sect1}

Let $\mathcal G=\{g_1,g_2, \ldots,g_n\}$ be a finite multiplicative group of 
order $n=|\mathcal G|$, with neutral element $g_1=1$. Let $\mathbb F$ be a field of characteristic $p$.
Finite fields of order $q=p^m$ are denoted as $\mathbb F_q$.
The group algebra $\mathbb F[\mathcal G]$ 
 of $\mathcal G$ over $\mathbb F$ consists of the formal sums
 \begin{equation}
   \label{dede1}
\sum_{j=1}^{n} \alpha_i g_i ~~~~\alpha_i \in \mathbb F, ~~g_i \in \mathcal G,
 \end{equation}
where the $\alpha_i$s are called coefficients. The sum in $\FF[\mathcal G]$ is 
 defined coefficientwise, that is the coefficients of the same $g_j$ are added following
 the addition in $\mathbb F$. The product is performed by applying the distributive law,
and the group elements are multiplied following the rule of $\mathcal G$.
The group algebra $\mathbb F[\mathcal G]$ is a vector space of dimension $n$
 over the field $\mathbb F$, and has the structure of an associative ring with identity. 
It is commutative if and only if $\G$ is commutative.
It emerges that the structure of the ideals 
depends on the group $\mathcal G$ and the field characteristic. 
If the field characteristic does not divide the group order, or is $0$, the group 
 ring is semisimple by Maschke's Theorem (see~\cite{maschke}). 
In addition, every ideal is principal and generated by an idempotent \cite{curtis}. 
If the field characteristic divides the group order, then the group ring 
is not in general semisimple \cite{navarro}. 
The problem of how to efficiently compute the dimension of an ideal is examined, via an approach using some elementary tools from representation theory.

\noindent
Every representation $D:\G\longrightarrow GL_m(\KK)$ of $\mathcal G$ 
over an extension field $\mathbb K$ of $\mathbb F$ induces a representation 
of the group algebra $\mathbb F[\mathcal G]$, which we denote again by $D$
\begin{equation}\label{dede1rap}
\begin{array}{rcl}
D: \FF[\G] & \longrightarrow & \M_m(\KK) \\ 
\sum_{j=1}^{n} \alpha_j g_j & \longmapsto & \sum_{j=1}^{n} \alpha_j D(g_j).
\end{array}
\end{equation}
Here, $\M_m(\KK)$ is the ring of $m\times m$ matrices with entries in $\KK$. 
In particular, a regular representation of $\G$ induces a representation of $\FF[\G]$ over $\mathbb F$. 

\noindent
An application to group codes will be considerd. A group code of length $n$ is a linear code, which is the image of an ideal $I\subseteq\FF[\G]$ via 
an isomorphism $\phi:\FF[\G]\rightarrow\FF^n$. In other words, an ideal of $\FF[\G]$ is a 
group code in $\FF^n$, for a given choice of a basis of $\FF[\G]$. Denote by $C=\phi(I)$ the group code 
corresponding to the ideal $I$, then $C$ is an $[n,k]$-code, where $k=\dim_{\FF} I$ (see~\cite{pless} for details). In this context, the computation of the dimension of  $I$ is one of the central problems. 

\noindent
The main results of the paper are contained in Section~\ref{main}:
 Proposition~\ref{theorem1} relates the dimension of a proper left/right ideal
 $I\subset\FF[\G]$ to the rank of a matrix constructed using a regular right/left
 representation of $\G$.
 Proposition~\ref{id_gen} and Theorem~\ref{corol} discuss how to compute an idempotent
 generator of a given left/right ideal, and show that the dimension of the left/right ideal
 that it generates can be obtained from the characteristic polynomial of the matrix from  
 Proposition~\ref{theorem1}. 
Corollary~\ref{charpoly} reduces the computation of the dimension of a principal left/right ideal to the computation of a 
characteristic polynomial, even in the case when the ideal is not generated by an idempotent. 
Section~\ref{exs} discusses the applications to coding theory and  presents some examples.

\section{Ideal dimension}\label{main}

Computing the dimension of a left/right ideal in $\mathbb F[\mathcal G]$ 
has both theoretical relevance and practical applications.
The proposed solution is based on the representation of
$\mathbb F[\mathcal G]$ induced by a regular representation of
$\mathcal G$. In the next proposition, the dimension of a left/right ideal is related
to the rank of a matrix built using the regular right/left representation of $\G$. 
The rest of the section discusses how to efficiently compute this rank;
in particular, it is related it to the characteristic polynomial of a suitable matrix.

\begin{proposition}\label{theorem1}
Let $\G=\{g_1,\ldots,g_n\}$ be a group. 
Let $\FF$ be a field, and $A=\FF[\G]$ be the group algebra generated by $\G$.
Let $f_1,\ldots,f_t\in A$, $f_i=\sum_{j=1}^n \alpha_{ij} g_j$ with $\alpha_{ij}\in\FF$.
Let $I=Af_1+\ldots+Af_t$ and $J=f_1A+\ldots+f_tA$. 
Let $\rho(f_1),\ldots,\rho(f_t)$ 
be the matrices corresponding to $f_1,\ldots,f_t$ in the representation of $A$ 
induced by the regular right representation $\rho$ of $\G$. Let 
$\lambda(f_1),\ldots,\lambda(f_t)$ be the matrices corresponding to $f_1,\ldots,f_t$ 
in the representation of $A$ induced by the regular left representation $\lambda$ 
of $\G$.
Define block matrices 
$$\rho(I)=\left[\begin{array}{c} \rho(f_1) \\ \vdots \\ \rho(f_t)
\end{array}\right]\;\;\; \mbox{ and }\;\;\;
\lambda(J)=\left[\begin{array}{c} \lambda(f_1) \\ \vdots \\ \lambda(f_t)\end{array}\right].$$ 
Then $$\dim_{\FF} I=\rk \rho(I)\;\;\; \mbox{ and }\;\;\; \dim_{\FF} J=\rk \lambda(J).$$
\end{proposition} 

\begin{proof}
For each $f=\sum_{i=1}^n a_ig_i\in A$, define $f^*=\sum_{i=1}^n a_ig_i^{-1}.$ 
We have $\rho(f^*)=(m_{ij})\in\M_n(\FF)$, where 
$$g_if=m_{i1}g_1+\cdots+m_{in}g_n.$$ 
This means that the entries
in the i-th row of $\rho(f_k^*)$ are the coefficients of $g_if_k$
in the $\FF$-basis $g_1,\ldots,g_n$ of $A$. The elements
$g_if_k$ for $i=1,\ldots,n$ and $k=1,\ldots,t$ generate 
$I$ as $\FF$-vector space, hence 
$$\dim_{\FF} I=\rk\left[\begin{array}{c} \rho(f_1^*) \\ \vdots \\ \rho(f_t^*)\end{array}\right].$$ 
The permutation of $\mathcal G$ that exchanges $g_i$ and $g_i^{-1}$ 
for each $i$ induces a permutation of the columns of $\rho(f^*)$ for all $f\in A$, 
which sends $\rho(f^*)$ to $\rho(f)$. Hence
 $$\dim_{\FF}I=\rk\left[\begin{array}{c} \rho(f_1^*) \\ \vdots \\ \rho(f_t^*)\end{array}\right]=\rk\left[\begin{array}{c} \rho(f_1) \\ \vdots \\ \rho(f_t)
\end{array}\right]=\rk\rho(I).$$
Similarly, for the right ideal $J=f_1A+\cdots+f_tA$, consider the regular left 
representation $\lambda$ of $\G$. The entries in the i-th row of $\lambda(f_k)$ are the
coefficients of $f_kg_i$ in the $\FF$-basis
$g_1,\ldots,g_n$. Since the elements $f_kg_i$ generate $J$ as an $\FF$-vector space, 
the rank of $\lambda(J)$ equals the dimension of $J$.
\end{proof}

\noindent
Computation of characteristic polynomials is straightforward, and can be used to compute 
the rank of a matrix. In this context, Proposition~\ref{theorem1} has interesting consequences. 
A preliminary result is first established; it is essentially contained in Theorem~24.2 of~\cite{curtis}, but is presented in a form useful here, together with a proof similar to that of
 \cite[Theorem 2.5.10, page 95]{polcino}.

\begin{lemma}
   \label{idemp}
Let $A$ be a semisimple ring. Let $I\subset A$ be a proper left (or right) ideal. 
Then $I=Ae$ (or $I=eA$), where $e\in A$ is an idempotent and $a=ae$ 
(or $a=ea$) for all $a\in I$.
Further, any idempotent generator of $I$ is of the form $e=1-\epsilon$, 
where $\epsilon$ is an idempotent and $e\epsilon=\epsilon e=0$. \\
If in addition $I$ is a two-sided ideal, then $I=(e)$ for some $e\in A$ idempotent.
 Moreover, $A=I\oplus J$ where $J=0:_A I$ is the annihilator of $I$, and $J=(\epsilon)$.
\end{lemma}

\begin{proof}
The proof is given for the case of left ideals; the proof for right ideals is analogous. 
Since $A$ is semisimple, $A=I\oplus J$, where $J\subset A$ is a left ideal. 
Write $1=e+\epsilon$, where $e\in I$ and $\epsilon\in J$. 
Multiplying the identity on the left by $e$, and using the fact that $J$ is a left ideal, 
we obtain that $$e\epsilon=e-e^2\in I\cap J=0.$$ Therefore $e=e^2\in I$ is an idempotent, 
and the same holds for $\epsilon\in J$.  \\
Clearly $Ae\subseteq I$. In order to show the reverse inclusion, observe that for any $a\in I$ 
we have $a\epsilon=a-ae\in I\cap J=0$, therefore $a=ae\in Ae.$ 
It follows that $I=Ae$ and $a=ae$ for all $a\in I$. This also shows that $J=A\epsilon$ and 
$a\epsilon=a$ for all $a\in J$. \\
If in addition $I$ is a two-sided ideal, then $A=I\oplus J$, where $J$ is also a two-sided ideal. 
Hence $I=(e)$ 
and $J=(\epsilon)$ with $e\epsilon=\epsilon e=0$. Therefore $IJ=JI=0$ and 
$$J\subseteq 0:_A I=\{a\in A\mid aI=Ia=0\}.$$ Conversely, let $a\in 0:_A I$. 
Then $a=ae+a\epsilon=a\epsilon\in J$. Therefore $J=0:_A I$, as claimed.
\end{proof}

\begin{remarks}\label{idemp_gen}~
\begin{enumerate}
\item Since in general a semisimple ring has elements which are not idempotent, 
there exist ideals that are generated by an element that is not idempotent. 
However, Lemma \ref{idemp} implies that they have another generator, which is idempotent. 

E.g., let $A=\FF_5[S_3]$ and let $f=1+(12)\in A$. Then $f^2=2f\neq f$; 
however, $Af=A(3f)$ and $e=3f$ is idempotent.

\item An idempotent matrix $M\in\M_n(\FF)$ has minimal polynomial $z,z-1$ or $z^2-z$. 
Hence its characteristic polynomial has the form $z^k(z-1)^{n-k}$ for some $0\leq k\leq n$.

In the previous example, following the notation of Proposition~\ref{theorem1}, we have 
$$\rho(f)=\left[\begin{array}{cccccc}
1 & 1 & 0 & 0 & 0 & 0 \\
1 & 1 & 0 & 0 & 0 & 0 \\
0 & 0 & 1 & 1 & 0 & 0 \\
0 & 0 & 1 & 1 & 0 & 0 \\
0 & 0 & 0 & 0 & 1 & 1 \\
0 & 0 & 0 & 0 & 1 & 1
\end{array}\right].$$ The matrix $\rho(f)$ is not idempotent, 
and has characteristic polynomial $z^3(z-2)^3$. 
Notice that $2\in\FF_5^*$ has order $4$, and $\rho(f)^4=3\rho(f)$ has characteristic 
polynomial $z^3(z-1)^3$. 
\end{enumerate}
\end{remarks}

Remark~\ref{idemp_gen}.2 raises the question of how to compute an idempotent generator of a 
given ideal $I$.
For example, let $I=Af$ be a left ideal with a given generator $f\in A$. 
Then $e-1\in\{a\in A\mid fa=0\}$. The problem of computing a right annihilator has a direct simple solution, since it reduces to a simple linear algebra problem.
In the next proposition, the case of left ideals is discussed; right ideals can be treated similarly.

\begin{proposition}\label{id_gen}
Let $D$ be a representation of $\G=\{g_1,\ldots,g_n\}$ such that $D(g_1),\ldots, D(g_n)$ are 
linearly independent.
Then $D$ induces an injective representation $D$ of $A=\FF[\G]$, and the right annihilators 
of $f\in A$ correspond via $D$ to the solutions over $\FF$ of the linear system 
$$D(f)\left(\sum_{i=1}^n x_i D(g_i)\right)=0$$   
in the variables $x_1 \ldots x_n$. In other words, solutions $(a_1,\ldots,a_n)\in\FF^n$ 
of the system correspond to right annihilators $a=a_1g_1+\ldots+a_ng_n$ of $f$. 
\end{proposition}

\begin{proof}
A representation $D:\G\rightarrow GL_m(\FF)$ induces a representation $D:A\rightarrow \M_m(\FF)$. $D$ is injective, since $D(g_1),\ldots, D(g_n)$ are linearly independent. Hence $D$ is an isomorphism between $A$ and its image, in particular $fa=0$ if and only if $D(fa)=D(f)D(a)=0$. Let $a=a_1g_1+\ldots+a_ng_n$, then $D(a)=a_1D(g_1)+\ldots+a_nD(g_n)$ and $D(f)D(a)=0$ if and only if $(a_1,\ldots,a_n)\in\FF^n$ is a solution of the linear system $$D(f)\left(\sum_{i=1}^n x_i D(g_i)\right)=0.$$
\end{proof}

\begin{corollary}
An idempotent generator $e$ of $I=Af$ (or $I=fA$) may be computed by solving a system of multivariate quadratic and linear equations over $\FF$.
\end{corollary}

\begin{proof}
The proof is given for the case of left ideals; the proof for right ideals is analogous. 
Let $e=x_1g_1+\ldots+x_ng_n+1$, where $x_1,\ldots,x_n$ are unknown. Proposition~\ref{id_gen} produces a system of linear equations in the coefficients of $e$. Imposing that $e^2=e$ we obtain a system of equations of degree two. The solutions to the system consisting of the linear and quadratic equations described correspond to the idempotent minimal generators of $I$.
\end{proof}

\noindent
The first consequence of Proposition~\ref{theorem1} may now be stated. Notice that if $p\nmid n=|\G|$, then $\FF[\G]$ is semisimple by Maschke's Theorem (see \cite{maschke}). In particular, by Lemma~\ref{idemp}, every left/right ideal of $A$ has an idempotent generator. In that case, the next result allows  the computation of the dimension of a left/right ideal to be reduced to computing a characteristic polynomial.

\begin{theorem}\label{corol}
Let $\G=\{g_1,\ldots,g_n\}$ be a group, $\FF$ be a field, and $A=\FF[\G]$.
Let $f=\alpha_1 g_1+\ldots+\alpha_n g_n\in A$ and let 
$I=Af$ be a proper left ideal (or $I=fA$ be a proper right ideal) of $A$. Let $F=\rho(f)$ be
the matrix  associated to $f$ in the regular right representation $\rho$ of $\G$ 
(or let $F=\lambda(f)$ be the matrix associated to $f$ in the regular left representation 
$\lambda$ of $\G$). Let $z^k g(z)$ be the characteristic polynomial of $F$, where $z\nmid g(z)$. 
Then $$n-k\leq\dim_{\FF} I\leq n-1.$$
If in addition $f$ is idempotent, then $\dim_{\FF} I=n-k$. 
\end{theorem}

\begin{proof}
It follows from Proposition~\ref{theorem1} that 
$$\dim_{\FF} I=\rk(F)\geq n-k.$$ 
Since $I$ is proper, $\dim_{\FF} I\leq n-1$. \\
If, in addition, $f$ is idempotent, then $F$ is an idempotent matrix, hence 
$$\ker F=\ker F^m\;\; \mbox{for any $m\geq 1$.}$$ 
Therefore, the algebraic and geometric multiplicities of $0$ as an eigenvalue of $F$ coincide, and 
$$\dim_{\FF} I=\rk(F)=n-k,$$ where the first equality again follows from Proposition~\ref{theorem1}.
\end{proof}

\noindent
In the case that $p\mid n$, there may be ideals of $\FF[\G]$ that have no idempotent generator. 
If $f$ is not idempotent, it is possible that $\rk(F)>n-k$.
Some examples follow in which $A$ is commutative and $I=(f)$ has $\dim_{\FF} I=\rk(F)>n-k$.
In particular, examples are given where $n-k=0$ and $\dim_{\FF} I=\ell$ for any $\ell$ 
which divides $n/p$, $p=\ch\FF$. 

\begin{example}
The notation of Theorem~\ref{corol} is followed. Let $\G$ be a cyclic group of order $n=\ell m$ 
and let $p$ be a prime, $p\mid  m$. 
Let $H$ be a subgroup of $\G$ of order $p$, and let $$f=\sum_{h\in H} h\in\FF[\G]$$ where $\FF$ is a field. 
For a suitable reordering of the elements of $\G$, the matrix $F\in\M_n(\FF)$ associated to $f$ in a regular representation of $\G$ 
is the diagonal block matrix of size $\ell\times\ell$
$$F=\left[\begin{array}{ccccc}
\mathbf{1} & \mathbf{0} & \mathbf{0} & \cdots & \mathbf{0} \\
\mathbf{0} & \mathbf{1} & \mathbf{0} & \cdots & \mathbf{0} \\
\vdots & \ddots & \ddots &  \ddots & \vdots \\
\mathbf{0} & \cdots & \mathbf{0} & \mathbf{1} & \mathbf{0} \\
\mathbf{0} & \mathbf{0} & \cdots & \mathbf{0} & \mathbf{1}
\end{array}\right]$$
with $m\times m$ blocks of 1s on the diagonal.
Here $\mathbf{0}$ and $\mathbf{1}$ denote the matrices of size $m\times m$ filled with $0$s and $1$s, respectively.
If $I=(f)\subset A$, then $$\dim_{\FF} I=\rk(F)=\ell$$ by Proposition~\ref{theorem1}.  
It is easy to check that the characteristic polynomial of $F$ is $z^{n-\ell}(z-m)^{\ell}$. 
Hence $k=n-\ell$ if $\ch\FF\neq p$, and $k=n$ if $\ch\FF=p$.
In particular, if $\ch\FF=p\mid n$, then $$\dim_{\FF} I=\ell>n-k=0.$$
\end{example}

\noindent
Combining Proposition~\ref{theorem1} with a result by Mulmuley \cite{mu87}, computing the dimension of a principal ideal can be reduced to computing a characteristic polynomial, even in the case that the generator of the ideal is not idempotent. 

\begin{corollary}\label{charpoly}
Let $\mathbb F$ be a field of arbitrary characteristic,
let $f\in A=\FF[\G]$ and let $I=Af$ (or $I=fA$).
Let $F$ be the matrix  associated to $f$ in the regular right representation 
(or in the regular left representation) of $\G$. Let 
$$M=F\;\;\; \mbox{if $\G$
is commutative,}$$ and let 
$$M=\left[\begin{array}{cc} 0 & F \\
F^t & 0
\end{array}\right] \;\;\; \mbox{if $\G$ is not commutative.}$$
Let $x$ be a transcendental element over $\FF$ and let $X$ be the
diagonal matrix with eigenvalues $1,x,\ldots,x^{m-1}$, where $m$ is
the size of the matrix $M$. Let $z^k g(z,x)$ be the
characteristic polynomial of the matrix $XM$, where $z\nmid g(z,x)$. Then 
$$\dim_{\FF} I=\left\{\begin{array}{ll} 
\deg_z g(z,x) & \mbox{if $\G$ is commutative,} \\
\deg_z g(z,x)/2 & \mbox{if $\G$ is not commutative.}
\end{array}\right.$$
\end{corollary}

\begin{remark}
If in Corollary~\ref{charpoly} we let $x$ be a random element in $\FF$ 
(or an element of a suitable algebraic extension field),  
 a faster randomized algorithm to compute the dimension of $I$ is obtained. 
This approach works well in practice, and the chance
of obtaining the correct result can be improved by 
repeating the computation for different random values of $x$. 
\end{remark}

\section{Application to coding theory and examples}\label{exs}

A group code is an $\FF_q$-linear code which is the image of an ideal $I\subseteq\FF_q[\G]$ 
via the isomorphism $\phi:\FF[\G]\rightarrow\FF^n$, which sends $g_i\in\G$ 
to the $i$th element of the canonical basis of $\FF^n$.
The length of $\phi(I)$ is equal to $n$, the cardinality of $\G$, and the dimension of the code is the dimension of $I$ over $\FF_q$. 
Thus, the method proposed for computing the dimension of ideals allows the dimension of
 group codes to be efficiently computed.
See~ \cite{ADR1,markov} for an overview of  group code
 properties, and ~\cite{elia} for an efficient encoding, and corresponding syndrome decoding. \\
The proposed method for computing the dimension of ideals will be illustrated by two simple examples. 
Observe that the right regular representation of a group $\G=\{g_1,\ldots,g_n\}$ 
can be easily obtained from a modified Cayley table:
The element in position $(i,j)$ in the table is $g_i^{-1}g_j$. Letting $g_1$ be the neutral element, 
the entries on the diagonal are all equal to $g_1$.
For a given $f=\sum_{i=1}^n x_i g_i \in A=\mathbb F[\mathcal G]$, the matrix $\rho(f)$ associated to $f$ in 
the right regular representation $\rho$ of $A$ is obtained by substituting $g_i$ by $x_i$
in the Cayley table for $i=1,\ldots,n$.

\begin{example}
Consider the Klein group 
$\mathcal K_4=\mathbb Z_2 \times \mathbb Z_2=\{e,\alpha, \beta, \alpha\beta \}$, 
where $e$ is the neutral element, $\alpha^2=\beta^2=1$, and $\alpha \beta =\beta \alpha$. 
Let $a,b,c,d\in\mathbb F$, then the matrix associated to 
$f=a e + b \alpha+ c \beta + d \alpha \beta\in\FF[\K_4]$ in the right regular representation $\rho$ of $\FF[\K_4]$ is
$$\rho(f)= \left[ \begin{array}{cccc}
        a & b & c &  d \\
        b & a & d &  c \\
        c & d & a &  b \\
        d & c & b &  a \\
      \end{array}  \right].
$$
If $\ch(\FF)\neq 2$, then $\rho(f)$ is diagonalizable
$$\rho(f)\sim \left[ \begin{array}{cccc}
        a+b+c+d & 0 & 0 &  0 \\
        0 & a-b-c+d & 0 &  0 \\
        0 & 0 & a+b-c-d &  0 \\
        0 & 0 & 0 &  a-b+c-d \\
      \end{array}  \right].
$$
Denote by $\phi_f(z)$ the characteristic polynomial of $\rho(f)$, then $$\phi_f(z)=(z-a-b-c-d)(z-a+b+c-d)(z-a-b+c+d)(z-a+b-c+d)=z^k g(z)$$ 
with $g(0)\neq 0$, and $$4-k=\deg g(z)=\rk\rho(f)=\dim_{\FF}(f).$$ Notice that every $f\in\FF[\K_4]$ is idempotent.
In this case $\FF[\K_4]$ is semisimple, and the diagonal form of $\rho(f)$ corresponds to the decomposition of $\FF[\K_4]$ 
as a direct sum of ideals 
$$\FF[\K_4]=(e+\alpha+\beta+\alpha\beta)\oplus (e-\alpha-\beta+\alpha\beta)\oplus (e+\alpha-\beta-\alpha\beta)\oplus (e-\alpha+\beta-\alpha\beta),$$ 
where each of the ideals appearing as direct summands is minimal and idempotent. \\
If $\ch(\FF)=2$, then $\rho(f)$ can be brought into the form
$$ \rho(f)\sim\left[ \begin{array}{cccc}
        a+b+c+d & 0       & 0 &  0 \\
         b      & a+b+c+d & 0 &  0 \\
         c      &  c+d    & a+d &  b+c \\
         d      &  c+d    & b+c &  a+d \\
      \end{array}  \right].$$
The characteristic polynomial of $\rho(f)$ is $$\phi_f(z)=(z+a+b+c+d)^4.$$
If $f$ is not invertible, then $a+b+c+d=0$, in particular $\phi_f(z)=z^4$.
Hence $\rk \rho(f)\in\{0,1,2\}$. In particular, $\rk\rho(f)=0$ if and only if $f=0$. 
Moreover, $\rk\rho(f)=1$ if and only if $a=b=c=d\neq 0$. In this case we have 
$(f)=(e+\alpha+\beta+\alpha\beta)$ and $\dim_{\FF}(f)=1$. Finally, if $a,b,c,d$ 
are not all equal, and either $b\neq 0$ or $b\neq c$ or $c\neq d$, then $\rk\rho(f)=2$. 
In this case, $f\in\{e+\alpha, e+\beta, e+\alpha\beta\}$ and $\dim_{\FF}(f)=2$.
Notice that if $\ch(\FF)=2$, then $\FF[\K_4]$ is not semisimple and every ideal equals its annihilator.
\end{example}

\begin{example}
Consider the symmetric group 
$S_3=\{ (1), (12), (13), (23), (123), (132)\}.$, where the notation $(.)$ denotes a cyclic substitution.\\
Every element $g\in\FF[S_3]$ is of the form
$$g=a(1)+b(12)+c(13)+d(23)+e(123)+f(132) ~~~  \mbox{where}~~a,b,c,d,e,f, \in \FF,$$
and is represented by the matrix
$$ \rho(g) = \left[ \begin{array}{cccccc}
        a & b & c & d & e & f \\
        b & a & e & f & c & d \\
        c & f & a & e & d & b \\
        d & e & f & a & b & c \\
        f & c & d & b & a & e \\
        e & d & b & c & f & a \\
      \end{array}  \right]  ~~,
$$
where $\rho$ denotes the right regular representation of $\G$. 
If $\ch(\FF)\neq 2,3$, then $\rho(f)$ can be brought into the following form:
\begin{equation}\label{diagg} \rho(g)\sim\left[ \begin{array}{cccccc}
        a+b+c+d+e+f &  0 &  0 &  0 &  0 &  0 \\
        0 &  a-b-c-d+e+f &  0 &  0 &  0 &  0 \\
        0 &  0 &  a-f &  e-f &  -c+d &  b-c \\
        0 &  0 &  -e+f &  a-e &  b-d &  c-d \\
        0 &  0 &  -c+d &  b-c &  a-f &  e-f \\
        0 &  0 &  b-d &  c-d &  -e+f &  a-e \\
      \end{array}  \right].
\end{equation}
Since $S_3$ has three equivalence clesses then $A$ can be decomposed as
 $$A=\FF[S_3]= I_1\oplus I_2\oplus I_3= I_1\oplus I_2\oplus J_1\oplus J_2 \,\,,$$ 
where $I_1,I_2,I_3$ are minimal two-sided ideals:
$$I_1= A((1)+(12)+(13)+(23)+(123)+(132)),~~ I_2= A((1)-(12)-(13)-(23)+(123)+(132)), $$
$$I_3=A((1)-(123))+A((12)-(23))+A((13)-(23))+A((123)-(132))=A(2\cdot (1)-(123)-(132)).$$
We have $\dim_{\FF} I_1=\dim_{\FF} I_2=1$ and $\dim_{\FF} I_3=4$. 
Moreover, $I_3$ can be decomposed as the direct sum of two left ideals:
$I_3= J_1\oplus J_2$ with 
$$ J_1 =A((1)+(12)-(23)-(123)),~~ J_2=A((1)-(12)+(23)-(132))$$ and $\dim_{\FF} J_1=\dim_{\FF} J_2=2$.
Using (\ref{diagg}) it is easy to check that in each case the dimension of the ideal 
is as predicted by Theorem~\ref{corol}.
It is  also easy to check that $I_1,I_2,J_1,J_2$ are minimal left ideals, hence any 
nonzero $I=Af\subseteq A$ is the sum of one or more among them.
In particular, $\dim_{\FF} I$ is the sum of the dimensions of the corresponding ideals, again 
as predicted by Theorem~\ref{corol}. \\
If $\ch(\FF)=2$, then 
$$\rho(g)\sim\left[ \begin{array}{cccccc}
        a+b+c+d+e+f &  0 &  0 &  0 &  0 &  0 \\
        0 &  a+e+f &  a+b+c+d+e+f &  0 &  0 &  0 \\
        0 &  a+b+c+d+e+f &  a+e & e+f &  b+c &  c+d \\
        0 &  0 &                       e+f & a+f &  b+d &  b+c \\
        0 &  0 &                       b+c & c+d &  a+e &  e+f \\
        0 &  0 &			    b+d & b+c &  e+f &  a+f \\
      \end{array}  \right].$$
Denoting by $I_1,I_2,I_3,J_1,J_2$ the same ideals as above, we have $I_1=I_2$, and an easy 
computation involving the above matrix yields
$$\dim_{\FF} I_1=1,~~\dim_{\FF} I_3=4,~~ \dim_{\FF} J_1=\dim_{\FF} J_2=2.$$
Notice that if $\ch(\FF)=2$, then $\FF[S_3]$ is no longer semisimple. \\
If $\ch(\FF)=3$, then $I_1,I_2\subset I_3$ where
$$I_3=A((1)-(123))+A((12)-(23))+A((13)-(23))+A((123)-(132))=A((1)+(123)+(132))$$ and 
$$\dim_{\FF} I_1=\dim_{\FF} I_2=1,~~\dim_{\FF} I_3=2.$$ Again $\FF[S_3]$ is not semisimple.
\end{example}  

\section*{Acknowledgments}  

The authors are grateful to David Conti for several useful discussions during the preparation of this paper. 
The second author was partially supported by the Swiss National Science Foundation under Grant no. 123393.

\end{document}